\documentclass[lettersize,journal]{IEEEtran}
\usepackage{amsmath,amsfonts}
\usepackage{algorithmic}
\usepackage{algorithm}
\usepackage{array}
\usepackage[caption=false,font=normalsize,labelfont=sf,textfont=sf]{subfig}
\usepackage{textcomp}
\usepackage{stfloats}
\usepackage{url}
\usepackage{verbatim}
\usepackage{graphicx}
\usepackage{cite}
\usepackage{booktabs, multirow}
\usepackage{amssymb}
\usepackage{amsmath, amsthm} 
\newtheorem{theorem}{Theorem}
\newtheorem{lemma}{Lemma}

\usepackage{enumitem}
\usepackage{subcaption}
\usepackage{color}

\begin{document}

\title{FedUTR: Federated Recommendation with Augmented Universal Textual Representation for Sparse Interaction Scenarios}

\author{Kang Fu,
Honglei Zhang,
Zikai Zhang,
Jundong Chen,
Xin Zhou,
Zhiqi Shen,~\IEEEmembership{Senior Member,~IEEE} \\
Dusit Niyato,~\IEEEmembership{Fellow,~IEEE}
and Yidong Li,~\IEEEmembership{Senior Member,~IEEE}

% <-this % stops a space

\thanks{Kang Fu, Honglei Zhang, Jundong Chen, and Yidong Li are with Key Laboratory of Big Data \& Artificial Intelligence in Transportation, Ministry of Education, and School of Computer Science and Technology, Beijing Jiaotong University, China (e-mail: \{kangfu, honglei.zhang, jundongchen, ydli\}@bjtu.edu.cn)}
\thanks{Zikai Zhang is with Carnegie Mellon University, Qatar (e-mail: zikaiz2@\allowbreak andrew.cmu.edu)}
\thanks{Xin Zhou, Zhiqi Shen, and Dusit Niyato are with College of Computing and Data Science, Nanyang Technological University, Singapore (e-mail: \{xin.zhou, zqshen, dniyato\}@ntu.edu.sg).}
% \thanks{Corresponding author: Yidong Li.}
}

% \markboth{IEEE TRANSACTIONS ON MULTIMEDIA,~Vol.~14, No.~8, August~2025}%


\maketitle

\begin{abstract}
Federated recommendations (FRs) have emerged as an on-device privacy-preserving paradigm, attracting considerable attention driven by rising demands for data security. Existing FRs predominantly adapt ID embeddings to represent items, making the quality of item embeddings entirely dependent on users' historical behaviors. However, we empirically observe that this pattern leads to suboptimal recommendation performance under high data sparsity scenarios, due to its strong reliance on historical interactions. To address this issue, we propose a novel method named \textbf{FedUTR}, which incorporates item textual representations as a complement to interaction behaviors, aiming to enhance model performance under high data sparsity. Specifically, we utilize textual modality as the universal representation to capture generic item knowledge, and design a Collaborative Information Fusion Module (CIFM) to complement each user's personalized interaction information. Besides, we introduce a Local Adaptation Module (LAM) that adaptively exploits the off-the-shelf local model to efficiently preserve client-specific personalized preferences. Moreover, we propose a variant of FedUTR, termed \textbf{FedUTR-SAR}, which incorporates a sparsity-aware resnet component to granularly balance universal and personalized information. The convergence analysis proves theoretical guarantees for the effectiveness of FedUTR. Extensive experiments on four real-world datasets show that our method achieves superior performance, with improvements of up to 59\% across all datasets compared to the SOTA baselines.
\end{abstract}


\section{Introduction}
\IEEEPARstart{R}{ecommendation} systems (RSs) aim to identify items that are likely to be of interest to users \cite{TMM2}. However, traditional RSs typically collect and centralize large volumes of user data on the server, which poses significant privacy risks, particularly under strict data protection regulations such as the General Data Protection Regulation (GDPR) \cite{GDPR}. 
To mitigate these privacy concerns, federated learning (FL) has been introduced into recommendation scenarios as a distributed learning paradigm that enables collaborative model training without sharing raw user data. FedAvg is the first FL framework \cite{FedAVG}, which learns a global model by iteratively aggregating locally trained updates via weighted averaging, and has become the foundational optimization backbone for a wide range of FL applications \cite{FCF, FL_activity_recognition, FL_TMM}.

Based on this highly emerging paradigm, \textbf{F}ederated \textbf{R}ecommendation\textbf{s} (FRs) adapt the FL optimization to recommendation tasks, enabling collaborative model learning across distributed users while preserving data locality \cite{FCF, TMM}. As illustrated in Fig. \ref{fig1}(a), the architecture based on user-item ID embeddings in FCF \cite{FCF} constitutes the prevailing framework for FRs in the client side \cite{lightFR, FedNCF, GPFedRec, FedRAP, FedCA}. 
These methods typically depict items based on their ID embeddings. However, the ID-based embedding methods overlook the intrinsic attributes of items, resulting in the quality of item embeddings being highly dependent on user-item interaction data. Consequently, the ID-based item embeddings may lead to suboptimal results when the data is highly sparse.
% Based on the FL paradigm, \textbf{F}ederated \textbf{R}ecommendation\textbf{s} (FRs) adapt federated optimization to recommendation tasks, enabling collaborative model learning across distributed users while preserving data locality \cite{FCF, TMM}. In the client-side, most methods perform recommendation through the score (e.g. similarity or inner product) between user embeddings and item embeddings. 
% Existing FR methods typically depict items based on their ID embeddings \cite{lightFR, GPFedRec}. As illustrated in Fig. \ref{fig1}(a), the architecture based on user-item ID embeddings in FCF constitutes the prevailing framework for FRs \cite{FCF}. However, the ID-based embedding methods overlook the intrinsic attributes of items, resulting in the quality of item embeddings being highly dependent on user-item interaction data. Consequently, the ID-based item embeddings may lead to suboptimal results when the data is highly sparse.

\begin{figure}
\centering
\includegraphics[width=\columnwidth]{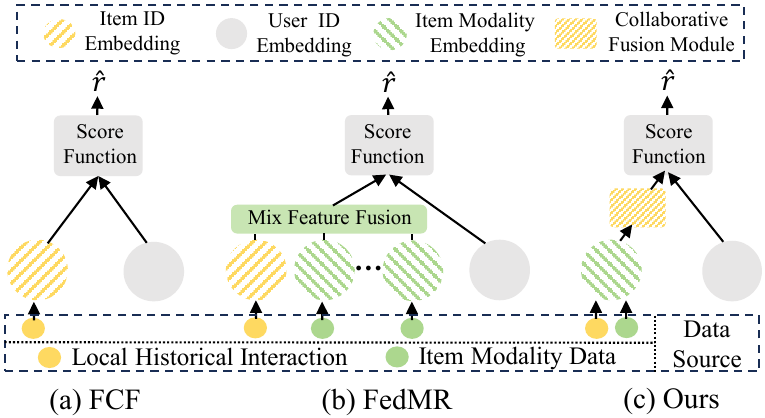}
\caption{Comparison of different client-side mechanisms in FRs. Our model augments universal textual modality on top of behavior interaction data in a parameter-efficient manner.}
\label{fig1}
\end{figure}

% 不应该叫协同信息融合模块，应该叫行为数据融合模块（IDFM），在通用文本数据的基础上引入行为信息。

To overcome this limitation, a natural solution is to enrich item representations with their associated modalities. Recently, advancements in Foundation Models (FM) have enabled more accurate extraction of \IEEEpubidadjcol modality data, which has already led to significant performance improvements in centralized recommendation scenarios \cite{BERT, CLIP, Llama, VBPR}. However, leveraging modality information in FRs poses significant challenges due to clients' constrained computational resources and storage capacity, which fundamentally restricts the complexity of on-device models. As shown in Fig. \ref{fig1}(b), FedMR pioneers the integration of multiple modalities (e.g., text, images) for items by performing mixed feature fusion in FRs \cite{FedMR}. However, such fusion strategy across multiple modality embeddings introduces substantial computational overhead, which significantly increases training latency. Besides, there inevitably exists substantial redundancy between textual and image modalities. Hence, recent work \cite{MultiVSingle} reveals that the increased capacity of multi-modal networks makes them more prone to overfitting, which ultimately causes them to yield inferior performance. This phenomenon becomes more significant in FRs, where the inherently high data sparsity further exacerbates this challenge.

\begin{figure*}[t]
\centering
\subfloat[\footnotesize Group1($\tau $=89.28\%)]{
  \includegraphics[width=0.185\textwidth]{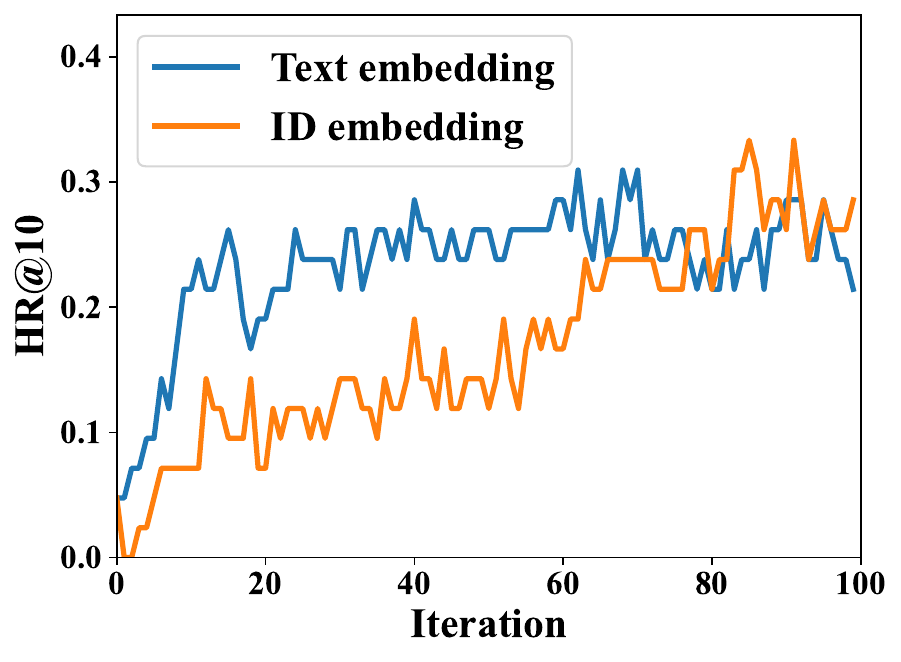}
  \label{fig2a}
}
\hfill
\subfloat[\footnotesize Group2($\tau $=91.98\%)]{
  \includegraphics[width=0.185\textwidth]{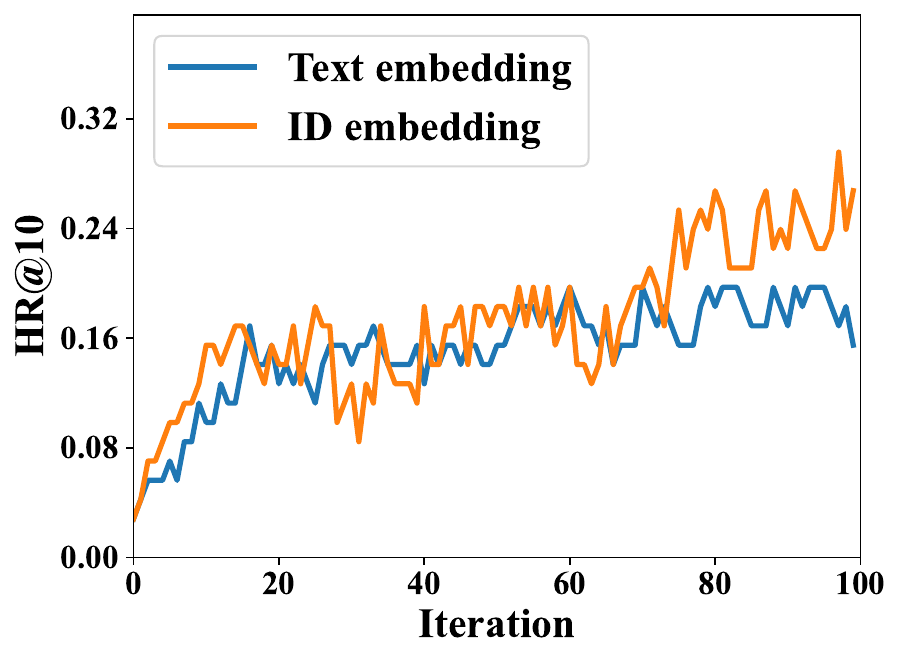}
  \label{fig2b}
}
\hfill
\subfloat[\footnotesize Group3($\tau $=94.52\%)]{
  \includegraphics[width=0.185\textwidth]{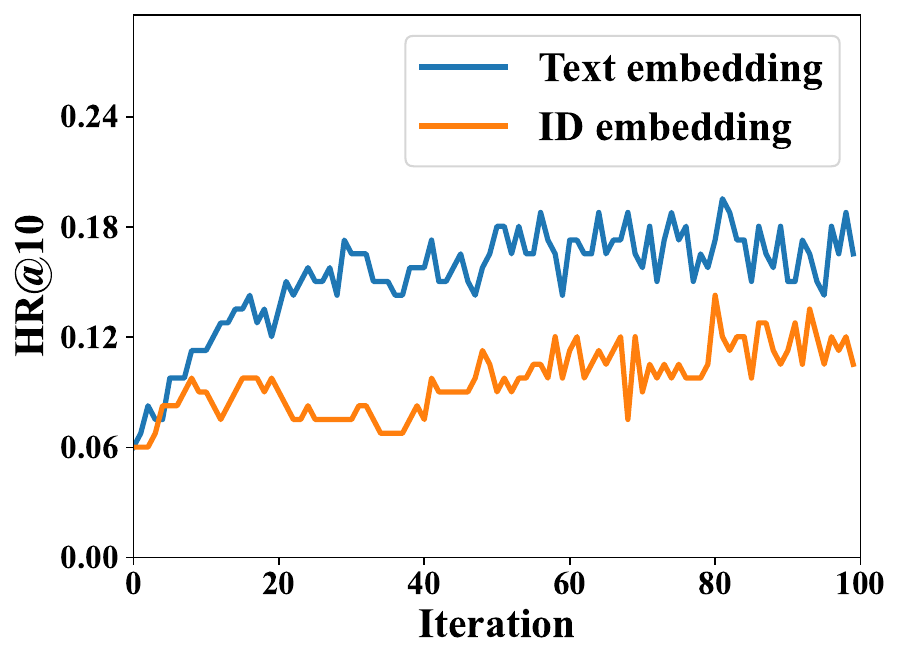}
  \label{fig2c}
}
\hfill
\subfloat[\footnotesize Group4($\tau $=96.08\%)]{
  \includegraphics[width=0.185\textwidth]{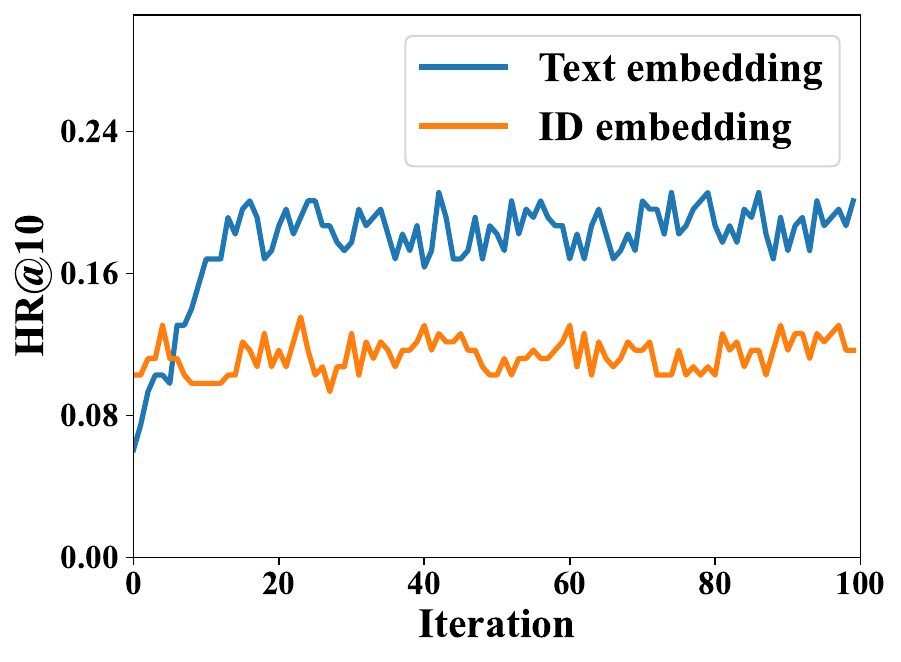}
  \label{fig2d}
}
\hfill
\subfloat[\footnotesize Group5($\tau $=99.50\%)]{
  \includegraphics[width=0.185\textwidth]{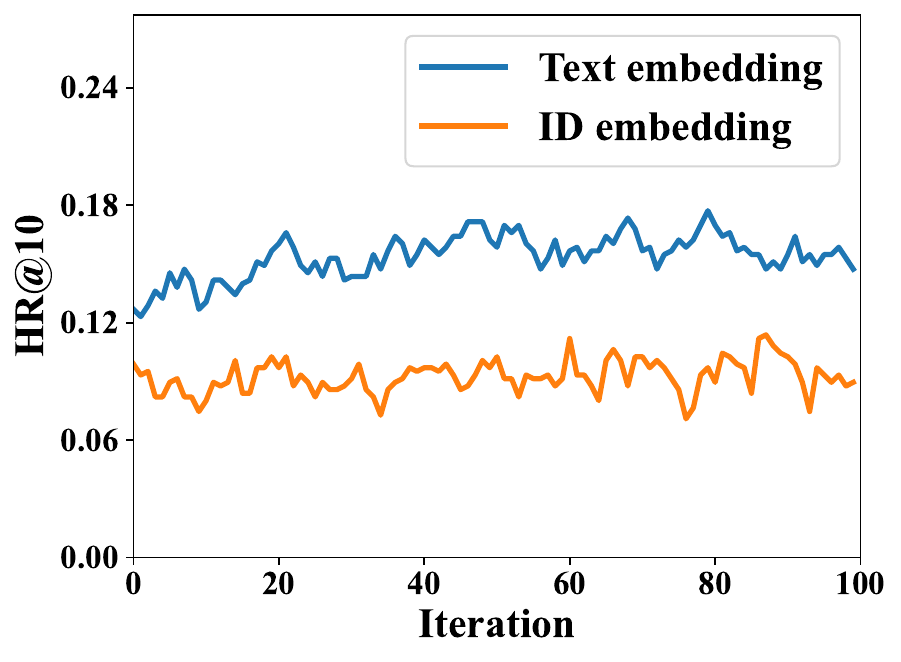}
  \label{fig2e}
}
\caption{The recommendation performance of text embeddings and ID embeddings across data subsets with varying sparsity levels. $\tau$ denotes the data sparsity level of the corresponding group. The x-axis represents training iterations, and the y-axis indicates HR@10.}
\label{fig2}
\end{figure*}

Based on this, we pose a fundamental question: Can we leverage one universal modality to enhance item representations while keeping the complexity of the model within acceptable limits for resource-limited clients in FRs? Motivated by this question, we design a case study for further validation. We use FCF as the backbone, partition users based on data sparsity in the Dance dataset \cite{NineRec}, and extract item text embeddings using foundation models (e.g., BERT) as a pioneering exploration. The model’s performance is then evaluated by using both ID embeddings and text embeddings as representations of the items, respectively. As shown in Fig. \ref{fig2}, the sparsity of the data gradually increases from Group 1 to 5. \textbf{We have the following observations and conclusions:}
\textbf{(1)} In relatively dense scenarios, text embeddings achieve performance comparable to ID embeddings, indicating that text embeddings are already effective even when ample interaction data is available.
\textbf{(2)} In highly sparse settings, ID embeddings suffer from severe performance degradation. In contrast, text embeddings can greatly enhance model performance, potentially serving as an effective complement to ID embeddings, demonstrating the robustness and universality of text modalities, particularly in sparse data regimes.

Based on the above observations, we propose a novel method suitable for sparse data scenarios in FRs, named \textbf{Fed}erated Recommendation with augmented \textbf{U}niversal \textbf{T}extual \textbf{R}epresentation \textbf{(FedUTR)}. Our approach integrates the complementary advantages of universal textual embeddings and personalized ID embeddings. To be specific, we design the Universal Representations Module (URM) that employs text embeddings as universal representations to depict intrinsic characteristics of items, and the Collaborative Information Fusion Module (CIFM) to capture personalized interaction information and universal textual knowledge across clients. As illustrated in Fig. \ref{fig1}(c), our method introduces only an additional CIFM compared to the conventional architecture in FCF, and its parameter size is negligible. Taking the KU dataset as an example, the parameter count increases only by 1.58\% compared to the FCF. In contrast to FedMR's explicit fusion strategy between ID and multiple modality-based embeddings, our framework achieves superior model performance with significantly fewer parameters. Besides, to efficiently preserve client-specific personalized preferences, we propose a Local Adaptation Module (LAM) to dynamically integrate global and off-the-shelf local models. Furthermore, we introduce a variant of FedUTR, named \textbf{FedUTR-SAR}, which designs a sparsity-aware module to adaptively balance the contributions of universal representations and behavior information according to each client's local sparsity. This variant is more suitable for scenarios where client-side computational resources are relatively sufficient and high performance is required. 

Our main contributions are as follows.

\begin{itemize}
    \item We empirically discover that the ID-based embedding approach fails to accurately capture item features in federated settings under high data sparsity. Therefore, we propose a novel framework, FedUTR, based on foundation models for sparse data scenarios.
    \item We introduce a URM to capture generic item knowledge, which can serve as a flexible plugin to be integrated with existing FRs. In addition, we design a CIFM to extract local interaction information and an LAM to preserve client-specific personalized preferences.
    \item We propose a more advanced sparsity-aware variant, FedUTR-SAR, which is tailored for clients with sufficient computational resources and stringent performance requirements.
    \item We provide a theoretical convergence analysis for the proposed scheme. Extensive experiments on four datasets demonstrate that our method consistently outperforms all baseline approaches.
\end{itemize}

% The rest of this paper is organized as follows. Section \ref{section 2} reviews the related work in the research domain. Section \ref{section 3} provides a detailed description of the proposed method. Section \ref{section 4} presents a theoretical analysis of the convergence of FedUTR. Section \ref{section 5} performs extensive experiments that address the five research questions mentioned above, demonstrating the effectiveness and superiority of our model. Finally, Section \ref{section 6} concludes the paper.

\section{Related work}
\label{section 2}
\subsection{Federated Recommendations}
FR is a privacy-preserving recommendation paradigm based on federated learning \cite{FL}, which performs recommendation accurately while ensuring user privacy and data security. FCF \cite{FCF} is the first federated recommendation algorithm built upon FedAvg \cite{FedAVG}, in which locally trained model updates are uploaded to the server and aggregated to form a global model. FedNCF \cite{FedNCF} further enhances recommendation performance through incorporating neural collaborative filtering (NCF) \cite{NCF} to capture higher-order nonlinear interactions. PFedRec \cite{PFedRec} removes the user embedding of each client and mimics the user’s decision logic through the score function. FedRAP \cite{FedRAP} preserves personalization and enhances communication efficiency by applying an additive model to item embedding. To achieve more effective personalized model aggregation, Fedfast \cite{Fedfast} improves training efficiency by employing active sampling and active aggregation mechanisms. GPFedRec \cite{GPFedRec} proposes a graph-based similarity-aware parameter aggregation approach that adaptively adjusts fusion weights based on topological correlations among clients in the global graph. However, existing mainstream methods primarily rely on ID embeddings to characterize items. FedMR \cite{FedMR} builds the federated multimodal recommendation framework by integrating modality features and ID embeddings, but this integration introduces substantial computational overhead, which poses significant challenges for industrial deployment. In this work, we propose FedUTR, a more effective method to enhance recommendation performance by leveraging modality information without significantly increasing model complexity.

\subsection{Foundation Models}
FM refers to models trained on broad data encompassing language, vision, and other domain corpora, which can be adapted to various downstream tasks (e.g., fine-tuning) \cite{FoundationSurvey}. The language model BERT \cite{BERT} pioneers bidirectional context encoding to reconstruct masked tokens via masked language modeling, while RoBERTa \cite{RoBERTa} enhances training efficiency through dynamic masking and larger batch sizes. Building on these, GPT-3(175B) demonstrates extreme success in language modeling by leveraging extensive text corpora training to align LLM capabilities with human intent \cite{GPT3, FoundationSurvey}, whereas LLaMA-65B \cite{Llama} achieves GPT-3 level performance on reasoning tasks with fewer parameters. In computer vision, ViT \cite{ViT} redefines image processing by applying pure transformer architectures to image patches, achieving excellent performance in image classification tasks. CLIP \cite{CLIP} bridges image-text understanding through contrastive objectives on image-text pairs, enabling zero-shot cross-modal transfer in various computer vision tasks. FM has shown remarkable feature extraction capabilities across various domains. In our FedUTR framework, we leverage FM to extract modality features of items, serving as initial universal representations.

\section{Methodology}
\label{section 3}
% In this section, we formalize the research problem of FRs and comprehensively present the proposed Federated Recommendation with Modality-Enhanced Representation (FedUTR) and its variant FedUTR-SAR. 

\subsection{Problem Formulation}
Let $\mathcal{U} = \{u_1, u_2, \ldots, u_n\}$ and $\mathcal{I} = \{i_1, i_2, \ldots, i_m\}$ denote the sets of users and items, respectively. Each user $u \in \mathcal{U}$ maintains a local dataset $\mathcal{D}_u = \{(u, i, r_{ui}) \mid i \in \mathcal{I}_u,\ r_{ui} \in \{0,1\}\}$, where $r_{ui}$ indicates the presence ($r_{ui}=1$) or absence ($r_{ui}=0$) of an interaction between user $u$ and item $i$. In the model training phase, each client $ u \in \mathcal{U} $ first trains its local model on the private interaction history $ \mathcal{D}_u $ by minimizing the binary cross-entropy loss:
\begin{equation}
\mathcal{L}_{rec}(\theta_u;\mathcal{D}_u ) =-\sum_{i \in \mathcal{I}_u}\log \hat{r}_{ui}-\sum_{i \in \mathcal{I}^{-}_u}\log (1 - \hat{r}_{ui}),
\label{eq1}
\end{equation}
where $\mathcal{I}_u$ and $\mathcal{I}^-_u$ denote the interacted positive item set and sampled negative item set of user $u$, respectively. $\hat{r}_{ui} = f(\theta_u;u,i)$ is the predicted interaction probability of the on-device model $f(\theta_u)$, and $r_{ui}$ is the true interaction from user $u$'s local dataset $\mathcal{D}_u$.

During the global model aggregation phase, the optimization objective over $n$ participating clients is formulated as:
\begin{equation}
    \min_{\{\theta : \theta_1,\theta_2,...,\theta_n\}}  \frac{1}{n} \sum_{u=1}^n \mathcal{L}_{rec}(\theta_u; \mathcal{D}_u),
    \label{eq2}
\end{equation}
where $\theta$ denotes the global model parameters, and $\mathcal{L}_{rec}(\theta_u; \mathcal{D}_u)$ represents the local loss computed on user $u$.

\subsection{Framework Overview}
 As illustrated in Fig. \ref{fig:framework}, FedUTR first leverages FM to extract item textual modality features as universal representations on the server, and then distributes these representations to clients as initial universal embeddings for items. Subsequently, each training round of FedUTR consists of the following steps: The CIFM enriches universal embeddings with local interaction knowledge and global collaborative information. These enriched item embeddings are then fed into the score function to obtain the final prediction scores. Then, the updated universal embeddings and CIFM parameters are uploaded to the central server for model aggregation. The aggregated parameters are further distributed to participating clients for the next training round. During this distribution phase, we adopt an LAM to dynamically integrate global collaborative information and local interaction knowledge, thereby effectively preserving user-specific personalized preferences. 
 
 In the following part, we provide the details of each component in the proposed FedUTR framework, following the algorithmic workflow.

\begin{figure*}[t]
\centering
\includegraphics[width=\textwidth]{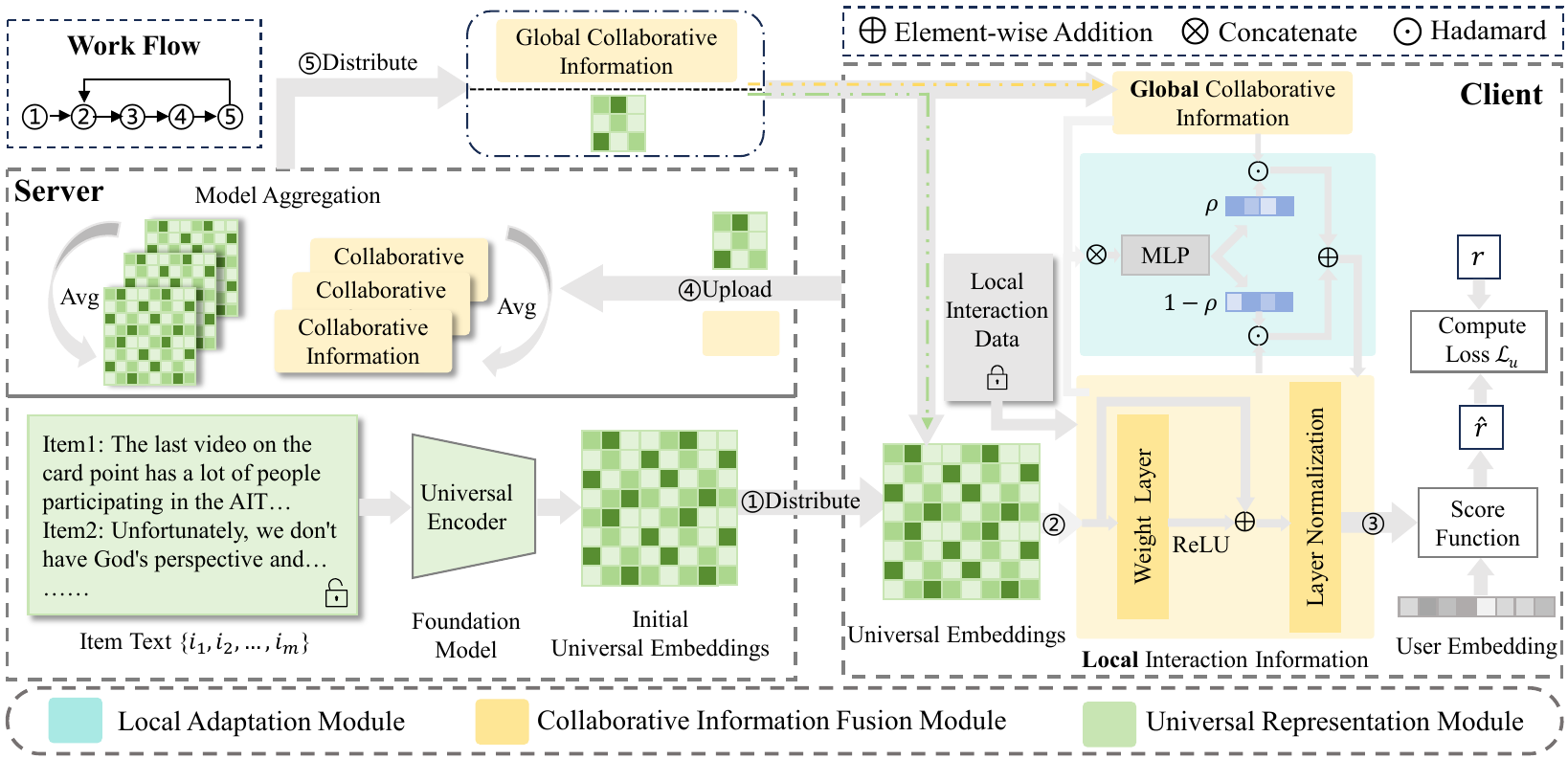} % Reduce the figure size so that it is slightly narrower than the column. Don't use precise values for figure width. This setup will avoid overfull boxes.
\caption{The overview framework of FedUTR. A foundation model extracts textual features as universal item embeddings. CIFM captures local interaction preferences, while LAM adaptively fuses global and local parameters. The server aggregates client parameters and redistributes the global model.}
% \caption{The overview framework of FedUTR. We deploy a foundation model to extract textual representations and distribute them to all participants as universal item embeddings. The CIFM is then designed to effectively capture users’ local interaction preferences. After each local training round, the server aggregates the uploaded parameters and redistributes the aggregated model to all clients. To preserve user-specific information, the LAM is applied to dynamically fuse global and local parameters.}
\label{fig:framework}
\end{figure*}

% 协同（Collaborative）不能滥用，协同一般指的是服务器中多个客户的的协同，或者说协同过滤中的协同，这里的CIFM改为IDFM，Global Augmented Information与Local Augmented Model（行为+文本），这里的augmented可以理解为在行为信息的基础上用文本模态来增强

% 还可以叫CIFM模块，只不过local beheaviour/collaborative information与global general information 

\subsection{Extracting Universal Representations}
The FM that is trained on vast corpora possesses rich general semantic knowledge and has the ability to effectively capture the modality information of items. To more efficiently leverage this generic knowledge and considering the resource constraints of clients in FRs, we deploy an FM on the server side. For simplicity, we select the parameter-efficient BERT model \cite{BERT} as the FM to extract item textual modality features for universal representation initialization. Specifically, given an item $i$, if we denote the FM as $\mathcal{F}$, the universal representation of item $i$ is formulated as:
\begin{equation}
    {\mathbf{E}_i} = \mathcal{F}\left(\left\{ { \texttt{[CLS]};t_1^i,t_2^i,...,t_k^i} \right\} \right),
    \label{eq:energy}
\end{equation}
where $\mathbf{E}_i \in \mathbb{R}^{d}$ is the initial universal representation of item $i$, $t_1^i$ denotes the first token in the textual modality of item $i$, and $k$ represents the total number of tokens in the textual sequence. \texttt{[CLS]} is a special token added in front of each input sequence, with the output vector at the \texttt{[CLS]} position serving as the holistic textual representation. Afterwards, the universal representations $\mathbf{E}$ are distributed to all clients as initial item universal embeddings (see Step 1 in Fig. \ref{fig:framework}). 
% Throughout subsequent training phases, the FM is no longer involved in training.

\subsection{Collaborative Information Fusion Module}
Given that the server-distributed initial universal embeddings contain only generic knowledge about items, we design a CIFM to enrich the universal representations with local interaction knowledge (see Step 2 in Fig. \ref{fig:framework}). In the CIFM, we introduce an MLP to capture local interaction knowledge from interaction data of users, and then employ residual connections to integrate the universal representations and interaction knowledge for preventing information dilution of item generic knowledge. Finally, we apply layer normalization to the fused embeddings to ensure consistent feature scaling:

\begin{equation}
    \mathbf{E}^F = \mathrm{LayerNorm}\left(\mathcal{G} ({\mathbf{E}}) \oplus {\mathbf{E}}\right),
    \label{eq:energy1}
\end{equation}
where $\mathcal{G}:\mathbb{R}^d \to \mathbb{R}^d$ is a single-layer MLP with ReLU activation, $\oplus$ represents element-wise addition, and $\mathbf{E},\mathbf{E}^F\in\mathbb{R}^{m\times d}$ denote the universal embeddings and the fused item embeddings, respectively. Each client computes the inner product between user embeddings and the non-interacted item embeddings in $\mathbf{E}^F$ to generate scores $\hat{r}$ for local recommendation (see Step 3 in Fig. \ref{fig:framework}).

Furthermore, considering the inherent client-side data sparsity in FRs, we locally incorporate a regularization term on the parameters of the CIFM into the optimization objective to mitigate overfitting risks. Formally, the regularized optimization objective is defined as:
\begin{equation}
    \mathcal{L}_{u}(\theta_u,\lambda; \mathcal{D}_u) = 
    \underbrace{\mathcal{L}_{\text{rec}}(\theta_u; \mathcal{D}_u)}_{\text{Recommendation Loss}} + \underbrace{\lambda \lVert \theta_u^{\rm{CIFM}} \rVert_1}_{\text{L1 Regularization}},
    \label{regularized_loss}
\end{equation}
where $\lambda$ is a hyperparameter controlling the regularization strength, $\theta_u^{\rm{CIFM}}$ represents the trainable parameters of the CIFM for user $u$, and $\Vert \cdot \rVert_1$ denotes an L1 norm. 

Upon completing a local training phase, all participating clients upload their locally updated universal embeddings and CIFM parameters to the central server for global parameter aggregation (see Step 4 in Fig. \ref{fig:framework}). The server then broadcasts the aggregated parameters back to all participating clients for subsequent training rounds or inference tasks. It is worth noting that the CIFM not only contains local interaction knowledge but also incorporates global collaborative information after global model aggregation. More details are provided in the next section.

\subsection{Local Adaptation Module}
The CIFM is designed to capture local interaction knowledge from user historical behaviors. This knowledge is then enhanced to collaborative information through global model aggregation, which improves its generalization capacity. However, if the local CIFM completely depends on the globally aggregated model, users will have access only to global collaborative information, losing their personalized interaction knowledge. Therefore, we propose an LAM that employs a gating mechanism to dynamically integrate the global and local parameters of CIFM, thereby effectively preserving user-specific preferences. In detail, we formalize the LAM as follows:
\begin{equation}
    \rho {\rm{ = \sigma}}\left( {{{\mathcal{G}}_{\rm{LAM}}}( {\theta _{g_{t-1}}^{\rm{CIFM}},\theta _{u_{t-1}}^{\rm{CIFM}}} )} \right),
    \label{eq:energy2}
\end{equation}

\begin{equation}
    \theta _{u_{t}}^{\rm{CIFM}} = \rho  \odot \theta_{g_{t-1}}^{\rm{CIFM}} {\rm{ + }}\left( {{\rm{1 - }}\rho } \right) \odot \theta _{u_{t-1}}^{\rm{CIFM}} ,
    \label{eq:energy3}
\end{equation}
where $\theta_{g_{t-1}}^{\rm{CIFM}}$ and $\theta _{u_{t-1}}^{\rm{CIFM}} $ denote the global parameters distributed by the server and the locally updated parameters of CIFM from the previous training round, respectively. $\sigma$ denotes the sigmoid activation function, ${\mathcal{G}}_{\rm{LAM}}$ denotes a parameterized network that generates the dynamic fusion weights $\rho \in \mathbb{R}^d$ for global and local parameters, $\odot$ represents the Hadamard product operator, and $\theta _{u_{t}}^{\rm{CIFM}} $ represents the fused CIFM parameters in the $t$-th training round. Note that the LAM is only applied to the CIFM, since the CIFM simultaneously captures local interaction knowledge and global collaborative information. In contrast, the universal embeddings, which capture generic item knowledge, remain consistent across all clients and do not require personalized model adjustments (see Step 5 in Fig. \ref{fig:framework}).

\subsection{FedUTR with Sparsity-Aware ResNet}
In our preliminary experiments, we observed that the importance of universal representations and interaction behavior information varies with different sparsity levels. To enhance each client’s ability to adaptively balance these two types of information under varying sparsity conditions, we design a sparsity-aware residual module based on FedUTR, as illustrated in Fig. \ref{fig:Resnet}.
Specifically, we introduce a sparsity-aware block that quantifies the local sparsity of each client by the number of interacted items. To mitigate the discrepancy across clients, we take the logarithm of this sparsity value as the final metric and feed it into the sparsity-aware block to generate dynamic weights. During the fusion process between the input and output of the residual block, we use these dynamic weights to adaptively balance the contributions of universal representations and interaction behavior information according to the client’s sparsity level. By replacing CIFM with the proposed sparsity-aware resnet module, we obtain a variant of FedUTR, named FedUTR-SAR.
% By replacing the original static residual block with the proposed sparsity-aware version, we obtain a variant of FedUTR, named $\textbf{FedUTR}$ with $\textbf{S}$parse-$\textbf{A}$ware $\textbf{R}$esnet ($\textbf{FedUTR-SAR}$).

\begin{figure}[t]
\centering
\includegraphics[width=\columnwidth]{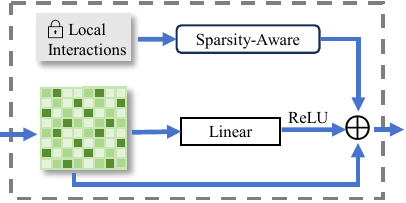} % Reduce the figure size so that it is slightly narrower than the column. Don't use precise values for figure width. This setup will avoid overfull boxes.
\caption{Sparsity-Aware ResNet Module. A sparsity-aware block is incorporated to measure local data sparsity and dynamically balance universal and personalized representations.}
\label{fig:Resnet}
\end{figure}

\subsection{Parameter Analysis}
To demonstrate the parameter efficiency of FedUTR, we compare the number of trainable parameters with representative baselines. As shown in Table \ref{tab:parameter}, while incorporating modality information, FedUTR requires only 29.97\%--41.53\% of the trainable parameters of FedMR, demonstrating its superior parameter efficiency.

% \begin{table}[h]
% \centering
% \caption{The number of trainable parameters.}
% \label{tab:parameter}
% \setlength{\tabcolsep}{6pt} % 压缩列间距（默认6pt）
% \begin{tabular}{lrrrrr}
% \toprule
% \textbf{Method} & \textbf{KU} & \textbf{Food} & \textbf{Dance} & \textbf{Movie} \\
% \midrule
%  \textbf{FCF} & 171,872 & 50,560 & 73,856 & 112,320 \\
%  \textbf{FedMR} & 420,420 & 177,796 & 224,388 & 301,316 \\
% \midrule
%  \textbf{FedUTR}  & 174,594 & 53,282 & 76,578 & 115,042 \\
% \bottomrule
% \end{tabular}
% \end{table}

\begin{table}[h]
\centering
\caption{The amount of trainable parameters (unit: Byte).}
\label{tab:parameter}
\setlength{\tabcolsep}{6pt} % 压缩列间距（默认6pt）
\begin{tabular}{lrrrrr}
\toprule
\textbf{Method} & \textbf{KU} & \textbf{Food} & \textbf{Dance} & \textbf{Movie} \\
\midrule
 \textbf{FCF} & 687,488 & 202,240 & 295,424 & 449,280 \\
 \textbf{FedMR} & 1,681,680 & 711,184 & 897,552 & 1,205,264 \\
\midrule
 \textbf{FedUTR}  & 698,376 & 213,128 & 306,312 & 460,168 \\
\bottomrule
\end{tabular}
\end{table}

% \subsection{Privacy-preserving Enhanced FedUTR}
% FedUTR, following the federated learning settings, preserves users’ private data locally during model training, thereby significantly reducing the risk of privacy leakage. To further enhance the privacy-preserving capabilities of FedUTR, we apply local differential privacy (LDP) to item embeddings before uploading the model to the server. Specifically, we achieve LDP by injecting zero-mean Laplace noise into client-side item embeddings, which can be formulated as:
% \begin{equation}
% \mathbf{E}_i = \mathbf{E}_i + \mathbf{X}, \ \text{where } \mathbf{X} \sim \text{Laplace}\left( 0, \delta \right), \ \delta = \Delta_f/\varepsilon,
% \label{eq:noise_mechanism}
% \end{equation}
% where $\mathbf{X} \in \mathbb{R}^{d}$ follows the Laplace distribution with zero mean and scale parameter $\delta$, $\Delta_f$ represents the L1-sensitivity of the embedding vectors, and $\varepsilon$ controls the privacy budget. An increase in $\delta$ leads to a decrease in the privacy budget $\varepsilon$, which in turn strengthens privacy protection. 

\section{Convergence Analysis of FedUTR}\
\label{section 4}
% \label{sec:convergence_compared}

In this section, we analyze the convergence behavior of FedUTR
by building upon the convergence analysis of FedAvg \cite{fedavg_convergence}, 
and focus on how the three proposed modules affect the convergence properties.
We adopt Assumptions 1--4 in \cite{fedavg_convergence}, summarized as follows:
(i) each $F_u$ is $L$-smooth;
(ii) each $F_u$ is $\mu$-strongly convex;
(iii) the stochastic gradient variance on each client is bounded by
$\sigma_u^2$;
(iv) the expected gradient norm is bounded by $G$;
Under these assumptions, FedAvg achieves an $\mathcal{O}(1/T)$ convergence rate.
Throughout the analysis, we define the optimal objective value as
$F^\star \triangleq F(\theta^\star)$.

\begin{lemma}[Effect of URM]
\label{lem:urm}
Consider FedUTR with the URM. If URM only modifies the initialization of model parameters and does not alter the local optimization procedure, then FedUTR with URM achieves the $\mathcal{O}(1/T)$ convergence rate.
\end{lemma}

\begin{proof}
URM initializes the item embeddings using modality features extracted by a foundation model on the server, while the local update rule and the global aggregation follow the standard FedAvg scheme. Therefore, the optimization trajectory of FedUTR with URM differs from FedAvg only in the initial model parameters.
Let $\theta^{\mathrm{URM}}_1$ denote the initialization.
According to Theorem~2 in \cite{fedavg_convergence}, we obtain
\begin{equation}
\begin{aligned}
\mathbb{E}[F(\theta_T)] - F^\star
\le
\frac{\kappa}{\gamma + T - 1}
\left(
\frac{2B}{\mu}
+
\frac{\mu \gamma}{2}
\mathbb{E}\|\theta_1^{\mathrm{URM}} - \theta^\star\|^2
\right),
\end{aligned}
\end{equation}
where $B$ is a constant depending on the stochastic gradient variance, data heterogeneity, and the boundedness of local gradients. 
\end{proof}

\begin{lemma}[Effect of CIFM]
\label{lem:cifm}
Consider FedUTR with the CIFM. Suppose that each client minimizes a composite objective consisting of a smooth loss and a convex regularizer, and performs proximal local updates. Then FedUTR with CIFM achieves the $\mathcal{O}(1/T)$ convergence rate.
\end{lemma}

\begin{proof}
With CIFM, the local objective on client $u$ is given by
\begin{equation}
\mathcal{L}_u(\theta)
=
F_u(\theta)
+
\lambda \|\theta_{\mathrm{CIFM}}\|_1,
\end{equation}
which defines a composite optimization problem with a smooth loss term and a convex regularizer.

By performing proximal local updates, this setting is equivalent to a proximal variant of FedAvg. Under the standard smoothness, bounded variance, and bounded gradient assumptions in \cite{fedavg_convergence}, the convergence analysis of FedAvg can be extended to this composite objective.
Specifically, we denote by $\sigma_{u,\mathrm{CIFM}}^2$ the variance bound of the stochastic gradient of $\mathcal{L}_u(\theta)$, by $\Gamma_{\mathrm{CIFM}} := F^\star - \sum_{u=1}^{N} p_u F_{u,\mathrm{CIFM}}^\star$ the induced heterogeneity measure, and by $G_{\mathrm{CIFM}}$ an upper bound on $\|\nabla \mathcal{L}_u(\theta)\|$. Accordingly, the constant $B$ in the convergence bound \cite{fedavg_convergence} is replaced by
\begin{equation}
B_{\mathrm{CIFM}}
=
\sum_{u=1}^{N} p_u^2 \sigma_{u,\mathrm{CIFM}}^2
+
6L \Gamma_{\mathrm{CIFM}}
+
8 (E-1)^2 G_{\mathrm{CIFM}}^2.
\end{equation}

We denote the resulting variance and heterogeneity bounds by $\sigma_{u,\mathrm{CIFM}}^2$ and $\Gamma_{\mathrm{CIFM}}$, respectively. Applying Theorem~2 in \cite{fedavg_convergence} yields
\begin{equation}
\mathbb{E}[F(\theta_T)] - F^\star
\le
\frac{\kappa}{\gamma + T - 1}
\left(
\frac{2B_{\mathrm{CIFM}}}{\mu}
+
\frac{\mu \gamma}{2}
\mathbb{E}\|\theta_1 - \theta^\star\|^2
\right).
\end{equation}

\end{proof}

\begin{lemma}[Effect of LAM]
\label{lem:lam}
Consider FedUTR with the LAM. If the local model update is given by a convex combination of the global and local parameters, then the local update drift is contractive. As a result, FedUTR with LAM achieves the $\mathcal{O}(1/T)$ convergence rate.
\end{lemma}

\begin{proof}
LAM updates the local model according to
\begin{equation}
\theta_u^{t}
=
\rho_u^{t} \odot \theta_g^{t-1}
+
(1-\rho_u^{t}) \odot \theta_u^{t-1},
\quad \rho_u^{t} \in (0,1).
\end{equation}
This update defines a convex combination between the global model
and the locally updated parameters.
By the non-expansiveness of convex combinations, we have
\begin{equation}
\|\theta_u^{t} - \theta_g^{t-1}\|
\le
(1-\rho_u^{t})
\|\theta_u^{t-1} - \theta_g^{t-1}\|,
\end{equation}
which shows that LAM induces a contraction mapping toward the global model and restricts the deviation of local models from the global parameter.

In the FedAvg convergence analysis, the dominant drift-related term
arises from multiple local updates and is upper bounded by
$8(E-1)^2 G^2$, where $G$ bounds the gradient norm over the entire
parameter space. Under LAM, local parameters are constrained to a shrinking neighborhood around the global model.
We therefore define the effective gradient bound $\tilde{G}$ as the
supremum of gradient norms over the parameter induced by LAM, which satisfies
$\tilde{G} \le G_{\mathrm{CIFM}}$.

Accordingly, the drift-related term admits a tighter upper bound
$8(E-1)^2 \tilde{G}^2$.
The resulting constant in the convergence bound becomes
\begin{equation}
B_{\mathrm{FedUTR}}
=
\sum_{u=1}^{N} p_u^2 \sigma_{u,\mathrm{CIFM}}^2
+
6L \Gamma_{\mathrm{CIFM}}
+
8 (E-1)^2 \tilde{G}^2,
\end{equation}
which satisfies $B_{\mathrm{FedUTR}} \le B_{\mathrm{CIFM}}$.

Under the contraction induced by LAM, this term admits a tighter upper bound
$8(E-1)^2 \tilde{G}^2$ with $\tilde{G} < G$. We have
\begin{equation}
\begin{aligned}
&\mathbb{E}[F(\theta_T)] - F^\star \\
&\le
\frac{\kappa}{\gamma + T - 1}
\left(
\frac{2B_{\mathrm{FedUTR}}}{\mu}
+
\frac{\mu \gamma}{2}
\mathbb{E}\|\theta_1 - \theta^\star\|^2
\right).
\end{aligned}
\end{equation}
\end{proof}

\begin{theorem}[Convergence of FedUTR]
\label{thm:FedUTR-final}
Under the standard smoothness, strong convexity, bounded variance, and bounded gradient assumptions in \cite{fedavg_convergence}, FedUTR achieves an $\mathcal{O}(1/T)$ convergence rate.
In particular, the following convergence bound holds for FedUTR
\begin{equation}
\begin{aligned}
\label{eq:FedUTR-final}
&\mathbb{E}[F(\theta_T)] - F^\star \\
&\le
\frac{\kappa}{\gamma + T - 1}
\left(
\frac{2B_{\mathrm{FedUTR}}}{\mu}
+
\frac{\mu \gamma}{2}
\mathbb{E}\|\theta_1^{\mathrm{URM}} - \theta^\star\|^2
\right).
\end{aligned}
\end{equation}
\end{theorem}

\begin{proof}
By Lemma~\ref{lem:urm}-\ref{lem:lam}, FedUTR with URM only modifies the initialization of the model parameters and does not affect the optimization procedure. FedUTR with CIFM introduces a convex regularizer and is optimized via proximal local updates. Under the same smoothness and bounded gradient assumptions, we reestablish the convergence bound in FedAvg under CIFM and obtain a bound in which the variance, heterogeneity, and gradient-related constants are modified accordingly. FedUTR with LAM induces a contraction toward the global model at each local update. As a result, local iterates are restricted to a smaller neighborhood around the global parameters, which yields a tighter gradient bound $\tilde{G} \le G_{\mathrm{CIFM}}$. Accordingly, the drift-related term $8(E-1)^2 G_{\mathrm{CIFM}}^2$ is reduced to
$8(E-1)^2 \tilde{G}^2$, leading to the constant $B_{\mathrm{FedUTR}}$ as following
\begin{equation}
\label{eq:B-FedUTR}
B_{\mathrm{FedUTR}}
=
\sum_{u=1}^{N} p_u^2 \sigma_{u,\mathrm{CIFM}}^2
+
6L \Gamma_{\mathrm{CIFM}}
+
8(E-1)^2 \tilde{G}^2 .
\end{equation}
Combining the above three Lemma and substituting $B_{\mathrm{FedUTR}}$
into the convergence bound in \cite{fedavg_convergence} completes the proof.

\end{proof}

\section{Experiments}
\label{section 5}
In this section, we conduct comprehensive experiments to answer the following research questions (RQ) to validate the effectiveness of FedUTR.
\begin{enumerate}[label=\textbf{RQ\arabic*},leftmargin=*,itemsep=0pt, parsep=0pt]
    \item Do FedUTR and its variant FedUTR-SAR outperform state-of-the-art federated baselines?
    \item How do the proposed URM, CIFM, and LAM contribute to the overall effectiveness of FedUTR?
    \item How do the key hyper-parameters influence model performance?
    \item Does the URM enhance the performance of existing methods as a plug-and-play component?
    \item Does the experimental result of FedUTR validate the convergence analysis conclusions?
\end{enumerate}

\subsection{Experiment Settings}
\begin{table}[h]
\centering
\caption{The statistical information of the used datasets.}
\label{tab:datasets}
\begin{tabular}{llllcc}
\toprule
Dataset & Users & Items & Interactions & Avg.I & Sparsity \\
\midrule
KU & 2034 & 5370 & 18519 & 9.11 & 99.83\% \\
Food & 6549 & 1579 & 39740 & 6.61 & 99.62\% \\
Dance & 10715 & 2307 & 83392 & 7.78 & 99.66\% \\
Movie & 16525 & 3509 & 115576 & 6.99 & 99.80\% \\
\bottomrule
\end{tabular}
\end{table}

We conduct comparative experiments with both centralized \cite{VBPR,BM3,MGCN} and federated baselines \cite{FedAVG, FCF, FedNCF, Fedfast, FedAtt, FedRAP, PFedRec, FedMR} to ensure comprehensive and impartial evaluation based on four datasets \cite{NineRec}. Detailed dataset statistics are shown in Table \ref{tab:datasets}.  
We evaluate performance by Hit Rate (HR) and Normalized Discounted Cumulative Gain (NDCG), with higher values denoting superior recommendation effectiveness.

\begin{table*}[t]
\centering
\caption{Performance Comparison of FedUTR and other baselines on four datasets. \textbf{CR} and \textbf{FR} represent centralized and federated recommendation model, respectively. The best results are highlighted in bold. The best baseline results are highlighted with an underline. \textbf{Improvement in FedUTR} indicates the performance improvement of our method relative to the strongest baseline, whereas \textbf{Improvement in FedUTR-SAR} represents the improvement of FedUTR-SAR over FedUTR.}
\label{tab:overall_results}
\begin{tabular}{llcccccccc}
\toprule
{} & \multirow{2}{*}{\textbf{Method}} & \multicolumn{2}{c}{\textbf{KU}} & \multicolumn{2}{c}{\textbf{Food}} & \multicolumn{2}{c}{\textbf{Dance}} & \multicolumn{2}{c}{\textbf{Movie}} \\
\cmidrule(lr){3-4} \cmidrule(lr){5-6} \cmidrule(lr){7-8} \cmidrule(l){9-10}
{} & & HR@10 & NDCG@10 & HR@10 & NDCG@10 & HR@10 & NDCG@10 & HR@10 & NDCG@10 \\
% {} & & H@10 & N@10 & H@10 & N@10 & H@10 & N@10 & H@10 & N@10 \\
\midrule
\multirow{3}{*}{\textbf{CR}} & 
\textbf{VBPR} & 0.2655 & 0.1555 & 0.0770 & 0.0375 & 0.0783 & 0.0394 & 0.0530 & 0.0273 \\
\multirow{3}{*}{} & 
\textbf{BM3} & 0.2478 & 0.1449 & 0.0843 & 0.0410 & 0.0837 & 0.0407 & 0.0603 & 0.0312 \\
\multirow{3}{*}{} & 
\textbf{MGCN} & 0.2581 & 0.1497 & 0.0893 & 0.0432 & 0.0819 & 0.0405 & 0.0602 & 0.0314 \\
\midrule
\multirow{8}{*}{\textbf{FR}} &
\textbf{FCF} & 0.1593 & 0.0648 & 0.0993 & 0.0437 & 0.0986 & 0.0421 & 0.1118 & 0.0515 \\
\multirow{8}{*}{} &
\textbf{FedAvg} & 0.1180 & 0.0515 & 0.1215 & 0.0563 & 0.1266 & 0.0616 & 0.1318 & 0.0600 \\
\multirow{8}{*}{} &
\textbf{FedNCF} & 0.1028 & 0.0430 & 0.1312 & 0.0576 & 0.1298 & 0.0594 & 0.1239 & 0.0543 \\
\multirow{8}{*}{} &
\textbf{Fedfast} & 0.0772 & 0.0349 & 0.0991 & 0.0435 & 0.1041 & 0.0445 & 0.1139 & 0.0522 \\
\multirow{8}{*}{} &
\textbf{FedAtt} & 0.1303 & 0.0655 & 0.1402 & 0.0663 & \underline{0.2589} & \underline{0.1317} & 0.1561 & 0.0742 \\
\multirow{8}{*}{} &
\textbf{FedRAP} & 0.1003 & 0.0453 & 0.1072 & 0.0500 & 0.1480 & 0.0702 & 0.1259 & 0.0581 \\
\multirow{8}{*}{} &
\textbf{PFedRec} & \underline{0.3564} & \underline{0.2710} & \underline{0.2117} & \underline{0.1002} & 0.2574 & 0.1238 & \underline{0.2246} & \underline{0.1146} \\
\multirow{8}{*}{} &
\textbf{FedMR} & 0.1028 & 0.0365 & 0.0151 & 0.0067 & 0.0099 & 0.0047 & 0.0096 & 0.0052 \\
\midrule
% \multirow{4}{*}{\textbf{Ours}} &
% \textbf{FedUTR} & \textbf{0.5693} & \textbf{0.3994} & \textbf{0.2622} & \textbf{0.1296} & \textbf{0.3477} & \textbf{0.1829} & \textbf{0.2551} & \textbf{0.1303} \\
\multirow{4}{*}{\textbf{Ours}} &
\textbf{FedUTR} & 0.5693 & 0.3994 & \textbf{0.2622} & 0.1296 & 0.3477 & 0.1829 & 0.2551 & 0.1303 \\
\cmidrule(lr){2-10}
& \textbf{Improvement} & 
59.74\%$\uparrow$ & 47.38\%$\uparrow$ & 23.85\%$\uparrow$ & 29.34\%$\uparrow$ & 34.30\%$\uparrow$ & 38.88\%$\uparrow$ & 13.58\%$\uparrow$ & 13.70\%$\uparrow$ \\
\cmidrule(lr){2-10}
&
\textbf{FedUTR-SAR} & \textbf{0.5777} & \textbf{0.4052} & 0.2590 & \textbf{0.1303} & \textbf{0.3624} & \textbf{0.1924} & \textbf{0.2605} & \textbf{0.1335} \\

& \textbf{Improvement} & 
1.48\%$\uparrow$ & 1.45\%$\uparrow$ & -1.22\%$\downarrow$ & 0.54\%$\uparrow$ & 4.23\%$\uparrow$ & 5.19\%$\uparrow$ & 2.12\%$\uparrow$ & 2.46\%$\uparrow$ \\
\bottomrule
\end{tabular}
\end{table*}

\begin{table*}[t]
\centering
\caption{Ablation study results. FedUTR ablation variants are denoted as: \textbf{FedUTR w/o URM} (Universal Representation Module removed), \textbf{FedUTR w/o CIFM} (Collaborative Information Fusion Module ablated), \textbf{FedUTR w/o LAM} (Local Adaptation Module excluded), and \textbf{FedUTR w/o Regular} (L1 regularization omitted from the objective). }
\label{tab:ablation_results}
\begin{tabular}{lcccccccc}
\toprule
\multirow{2}{*}{\textbf{Method}} & \multicolumn{2}{c}{\textbf{KU}} & \multicolumn{2}{c}{\textbf{Food}} & \multicolumn{2}{c}{\textbf{Dance}} & \multicolumn{2}{c}{\textbf{Movie}} \\
\cmidrule(lr){2-3} \cmidrule(lr){4-5} \cmidrule(lr){6-7} \cmidrule(l){8-9}
\multirow{2}{*}{} & HR@10 & NDCG@10 & HR@10 & NDCG@10 & HR@10 & NDCG@10 & HR@10 & NDCG@10 \\
% \multirow{2}{*}{}  & H@10 & N@10 & H@10 & N@10 & H@10 & N@10 & H@10 & N@10 \\
\midrule
\textbf{FedUTR w/o URM} & 0.3535 & 0.2739 & 0.1962 & 0.0946 & 0.2644 & 0.1321 & 0.2263 & 0.1168 \\
\textbf{FedUTR w/o CIFM} & 0.5501 & 0.3803 & 0.2478 & 0.1201 & 0.3281 & 0.1695 & 0.2336 & 0.1171 \\
\textbf{FedUTR w/o LAM} & 0.5383 & 0.3802 & 0.2567 & 0.1294 & 0.3417 & 0.1810 & 0.2533 & 0.1282 \\
\textbf{FedUTR w/o Regular} & 0.5688 & 0.3993 & 0.2512 & 0.1239 & 0.3342 & 0.1767 & 0.2422 & 0.1234 \\
\textbf{FedUTR} & \textbf{0.5693} & \textbf{0.3994} & \textbf{0.2622} & \textbf{0.1296} & \textbf{0.3477} & \textbf{0.1829} & \textbf{0.2551} & \textbf{0.1303} \\
\bottomrule
\end{tabular}

\end{table*}

\subsection{Overall Performance (RQ1)}

We compare the performance of baselines and FedUTR on four datasets. The experimental results are presented in Table \ref{tab:overall_results}, from which we have the following observations: (1) FedUTR outperforms the three centralized multimodal recommendation methods across all datasets. In contrast to centralized approaches, which share a single set of parameters for all users, FedUTR integrates personalized modules on the client side, leading to more accurate and personalized recommendations. 2) Our method consistently achieves state-of-the-art performance among all federated recommendation baselines. In our experiments, all four selected datasets exhibit high sparsity, and traditional FR approaches that rely solely on interaction data for item representation consequently demonstrate inferior performance. In contrast, FedUTR additionally incorporates universal representations to capture intrinsic item features, rather than depending entirely on historical interactions to characterize items, thereby enhancing its representational capacity and improving recommendation accuracy. (3) Compared to FedMR, a multimodal federated recommendation method, our approach achieves better performance while significantly reducing the model’s parameter size. (4) FedUTR-SAR achieves overall better performance than FedUTR in most cases and consistently outperforms all baseline methods. However, compared with FedUTR, the introduction of the sparsity-aware module incurs additional computational cost during training. Although FedUTR-SAR demonstrates superior performance, the improvement over FedUTR is not substantial and is not consistently observed across all scenarios. Therefore, FedUTR-SAR is more suitable for environments where client-side computational resources are relatively sufficient and high performance is required. In contrast, FedUTR can achieve a better trade-off between performance and computational complexity under resource-constrained conditions.

\subsection{Ablation Study (RQ2)}
The ablation study consists of two parts. In the first part, we conduct ablations on the main modules of FedUTR to evaluate the contribution of each component to the overall performance. In the second part, we further analyze the capability of the URM and CIFM modules in capturing universal and personalized information, respectively.

\begin{table*}[t]
\centering
\caption{Plug-and-Play compatibility verification results. \textbf{w/ URM} denotes integrating URM into backbone. \textbf{Improvement} denotes the performance increase achieved by integrating URM compared to the original method without URM integration.}
\label{tab:Plug_results}
% \normalsize
% \setlength{\tabcolsep}{5pt} % 压缩列间距（默认6pt）
\begin{tabular}{lccccccccc}
\toprule
\multirow{2}{*}{\textbf{Method}}& \multicolumn{2}{c}{\textbf{KU}} & \multicolumn{2}{c}{\textbf{Food}} & \multicolumn{2}{c}{\textbf{Dance}} & \multicolumn{2}{c}{\textbf{Movie}} \\
\cmidrule(lr){2-3} \cmidrule(lr){4-5} \cmidrule(lr){6-7} \cmidrule(l){8-9}
& HR@10 & NDCG@10 & HR@10 & NDCG@10 & HR@10 & NDCG@10 & HR@10 & NDCG@10 \\
\midrule
\textbf{FCF} & 0.1593 & 0.0648 & 0.0993 & 0.0437 & 0.0986 & 0.0421 & 0.1118 & 0.0515 \\
\textbf{FCF w/ UR} & \textbf{0.2984} & \textbf{0.1502} & \textbf{0.1545} & \textbf{0.0726} & \textbf{0.1532} & \textbf{0.0735} & \textbf{0.1317} & \textbf{0.0605} \\
\textbf{Improvement} & 87.32\% $\uparrow$ & 131.79\%$\uparrow$  & 55.59\%$\uparrow$ & 66.13\%$\uparrow$ & 55.38\%$\uparrow$ &74.58\%$\uparrow$ & 17.80\%$\uparrow$ & 17.48\%$\uparrow$ \\
\midrule
\textbf{FedNCF} & 0.1028 & 0.0430 & 0.1312 & 0.0576 & 0.1298 & 0.0594 & 0.1239 & 0.0543\\
\textbf{FedNCF w/ UR} & \textbf{0.4582} & \textbf{0.2870} & \textbf{0.1820} & \textbf{0.0870} & \textbf{0.2703} & \textbf{0.1356} & \textbf{0.1737} & \textbf{0.0836} \\
\textbf{Improvement} & 345.72\%$\uparrow$ & 567.44\%$\uparrow$ & 38.72\%$\uparrow$ & 51.04\%$\uparrow$ & 108.24\%$\uparrow$ & 128.28\%$\uparrow$ & 40.19\%$\uparrow$ & 53.96\%$\uparrow$ \\
\midrule
\textbf{FedRAP} & 0.1003 & 0.0453 & 0.1072 & 0.0500 & 0.1480 & 0.0702 & 0.1259 & 0.0581 \\
\textbf{FedRAP w/ UR} & \textbf{0.3732} & \textbf{0.2838} & \textbf{0.165}1 & \textbf{0.0799} & \textbf{0.1941} & \textbf{0.0924} & \textbf{0.1555} & \textbf{0.0724} \\
\textbf{Improvement} & 272.08\%$\uparrow$ & 526.49\%$\uparrow$ & 54.01\%$\uparrow$ & 59.80\%$\uparrow$ & 31.15\%$\uparrow$ & 31.62\%$\uparrow$ & 23.51\%$\uparrow$ & 24.61\%$\uparrow$ \\
\bottomrule
\end{tabular}

\end{table*}

\begin{figure}[t]
    \centering
    % 第二行
    \subfloat[\footnotesize Universal Representation]{
            \includegraphics[width=0.23\textwidth]{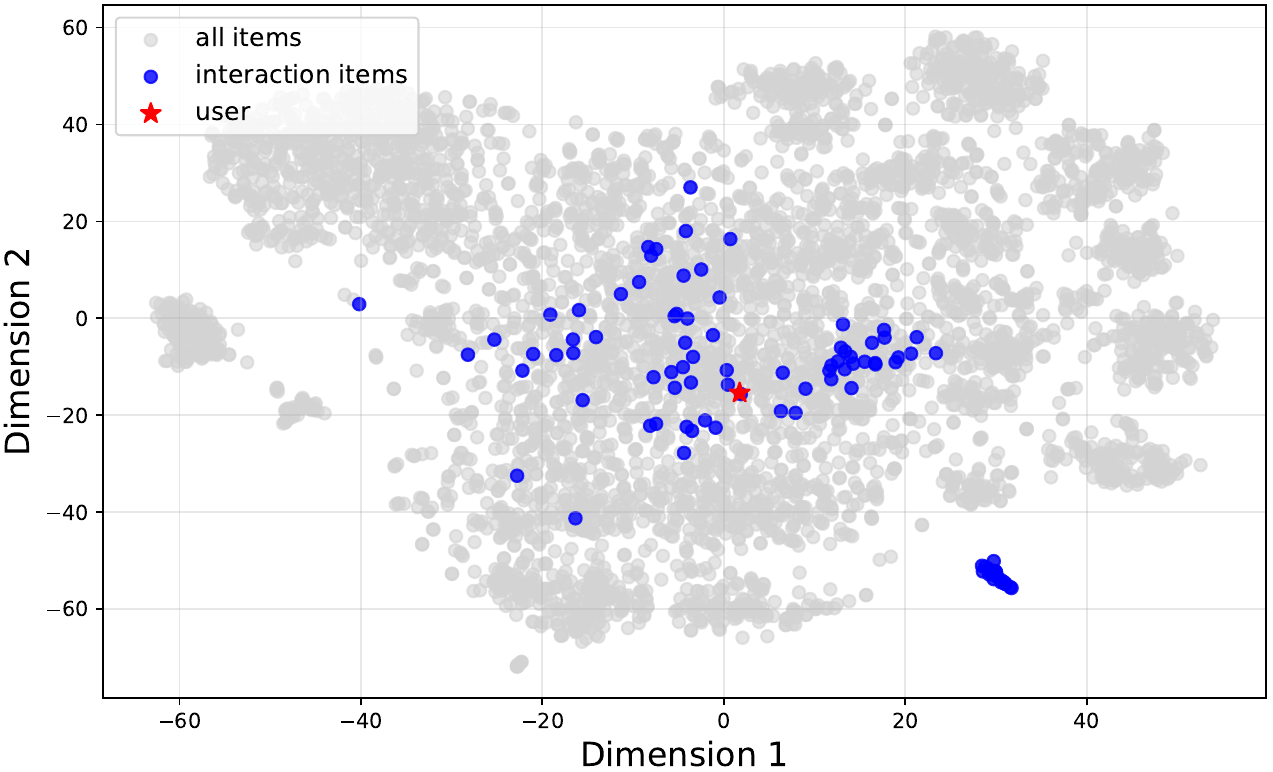}
    }
    \subfloat[\footnotesize Representation of CIFM]{
        \includegraphics[width=0.23\textwidth]{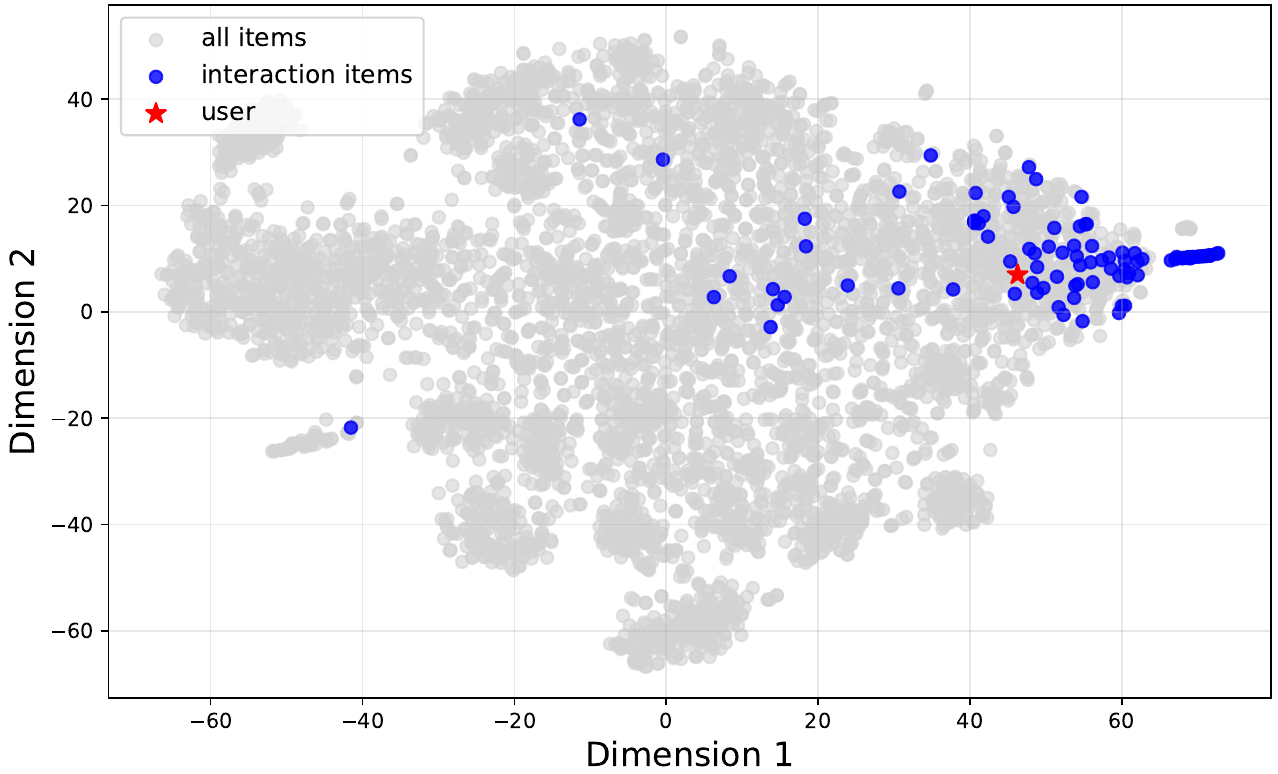}
    }
    \\
    \subfloat[\footnotesize CS with interactions]{
        \includegraphics[width=0.23\textwidth]{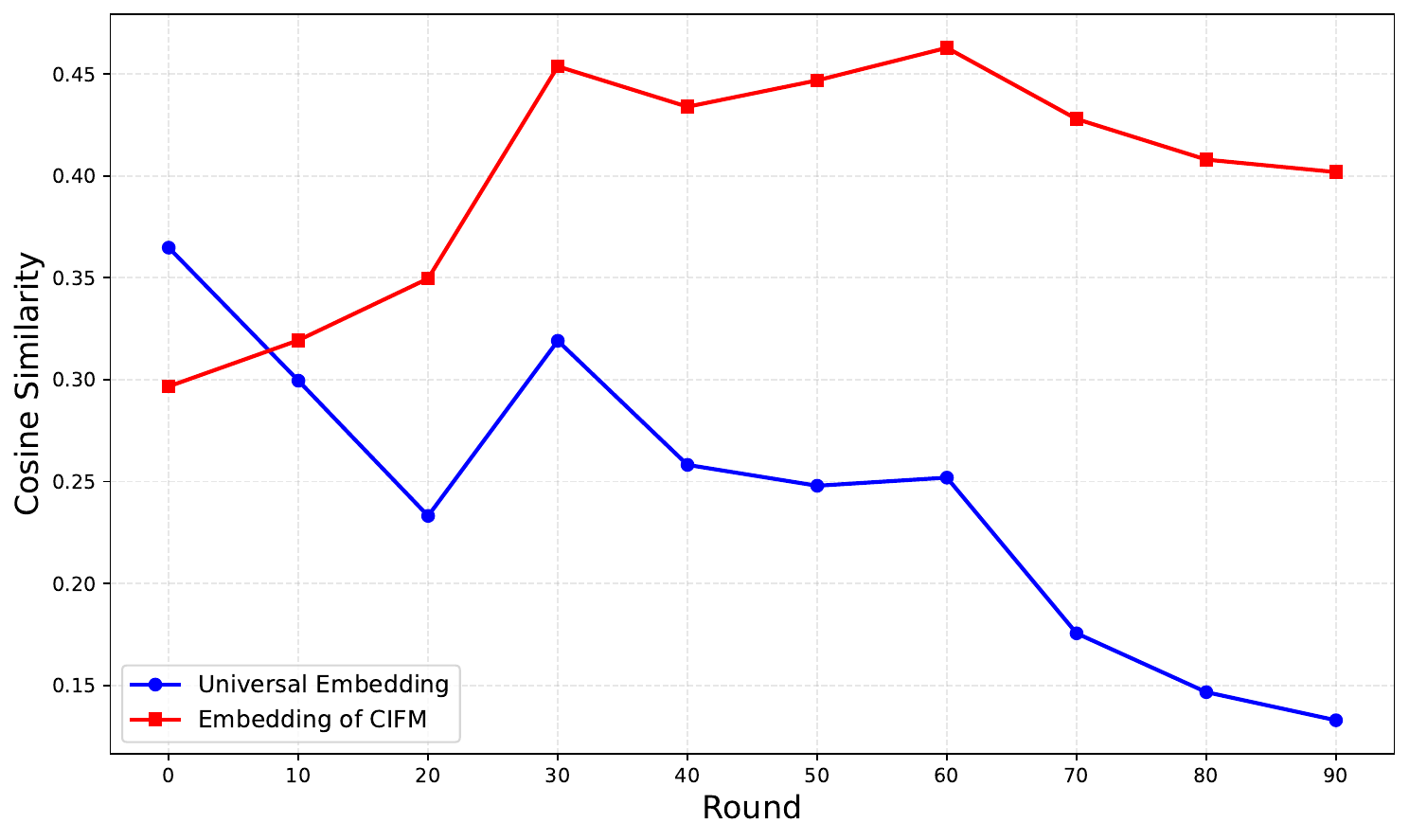}
    }
    \subfloat[\footnotesize CS with non-interactions]{
        \includegraphics[width=0.23\textwidth]{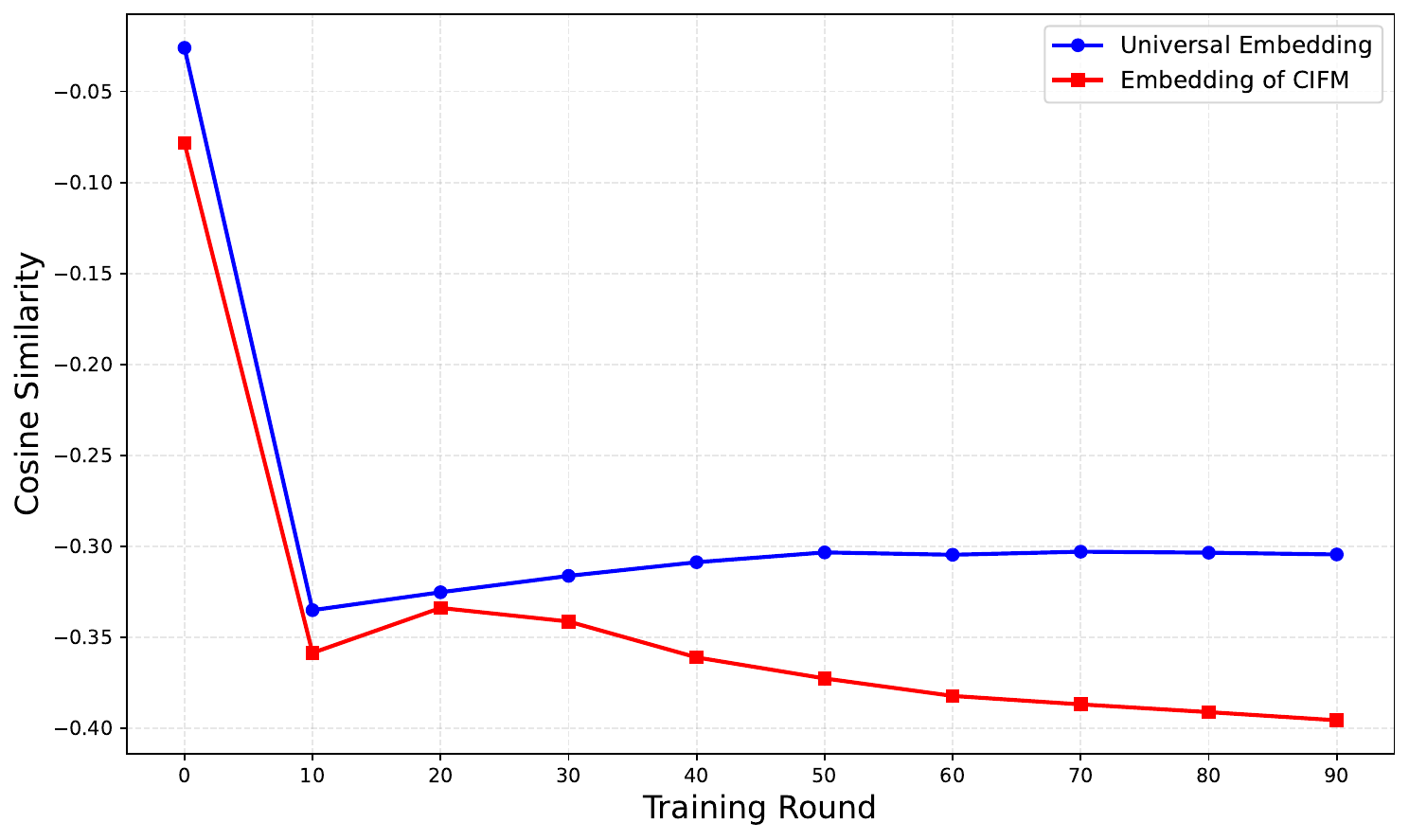}
    }
\caption{
Visualization of universal representations and CIFM-fused item representations for a user in the KU dataset. The first row shows item embeddings, while the second row illustrates the cosine similarity between the user and interacted/non-interacted items over training rounds.
}
\label{fig:visualization}
\end{figure}

First, we conduct an ablation study to investigate the impact of different modules on overall performance. Based on the experimental results reported in Table \ref{tab:ablation_results}, we observe that (1) the URM exerts the most significant impact on FedUTR's performance, with its exclusion leading to the most pronounced performance drop. In FedUTR, universal representations are introduced to complement the limitations of traditional methods that rely entirely on interaction data, particularly in sparse scenarios. When this module is removed, the model reverts to a conventional ID-based embedding approach, resulting in a significant performance degradation. (2) The removal of the CIFM and LAM also leads to varying degrees of performance degradation in the model. The CIFM is designed to capture local interaction knowledge and collaborative information. Our preliminary experiments have revealed that collaborative information provides better performance only in low data sparsity. Hence, the CIFM has a more significant effect on clients with denser interactions, yielding improvements that are relatively smaller compared to the URM. The LAM is designed to preserve user-specific preferences on the client. Although the LAM, which operates on the CIFM to preserve personalized information on the client side, contributes to performance improvement, its effectiveness is constrained by the inherent capabilities of the CIFM. (3) The regularization term exhibits a relatively minor impact (almost negligible) on the KU dataset, while demonstrating more pronounced effects in the other three datasets. This discrepancy primarily stems from differences in dataset characteristics, with a detailed analysis provided in section \ref{hyper-parameter-analysis}.

Second, to provide a more intuitive demonstration of the distinct roles of universal representations (URM) and collaborative information (CIFM) in client-side model training, we randomly selected a user for experimental analysis, as shown in Fig. \ref{fig:visualization}. 
The Fig. \ref{fig:visualization} (a) and (b) depict the item universal representations and the item representations fused with local interaction information via the CIFM layer, respectively. 
From the results, we observe that in the distribution of universal representations, items that interacted with the user are relatively scattered and lack strong personalization. In contrast, after incorporating local personalized behavior information, interacted item embeddings exhibit a clear tendency to cluster around the corresponding user embedding, highlighting the effect of CIFM.

To quantitatively illustrate this difference, we compute the cosine similarity (CS) between users and items. The Fig. \ref{fig:visualization} (c) and (d) show the similarity between users and interacted/non-interacted items. The horizontal axis represents training rounds, and the vertical axis represents cosine similarity. 
The experimental results reveal two key observations: (1) For items that users have interacted with, the similarity between users and items in the fused representation space is significantly higher than in the universal representation space; (2) For items that users have not interacted with, the similarity in the fused space is lower than in the universal representation space. These findings further confirm that the universal representations (UR) capture generic knowledge across clients, while the CIFM effectively complements personalized behavior information for each client.

% \subsection{Personalization about CIFM}

% \begin{figure*}[htbp]
%     \centering
%     % 第一行
%     \subfloat[user 354]{
%         \includegraphics[width=0.24\textwidth]{./image/354avg_interacted_euclidean_comparison.pdf}
%     }
%     \subfloat[user 611]{
%         \includegraphics[width=0.24\textwidth]{./image/611avg_interacted_euclidean_comparison.pdf}
%     }
%     \subfloat[user 2025]{
%         \includegraphics[width=0.24\textwidth]{./image/2025avg_interacted_euclidean_comparison.pdf}
%     }
%     \subfloat[user 2027]{
%         \includegraphics[width=0.24\textwidth]{./image/2027avg_interacted_euclidean_comparison.pdf}
%     }

%     \\ % 换行

%     % 第二行
%     \subfloat[user 354]{
%             \includegraphics[width=0.24\textwidth]{./image/354avg_non_interacted_euclidean_comparison.pdf}
%     }
%     \subfloat[user 611]{
%         \includegraphics[width=0.24\textwidth]{./image/611avg_non_interacted_euclidean_comparison.pdf}
%     }
%     \subfloat[user 2025]{
%         \includegraphics[width=0.24\textwidth]{./image/2025avg_non_interacted_euclidean_comparison.pdf}
%     }
%     \subfloat[user 2027]{
%         \includegraphics[width=0.24\textwidth]{./image/2027avg_non_interacted_euclidean_comparison.pdf}
%     }

%     \caption{The evolution of cold-start item embeddings during training. Early in training, these embeddings deviate from the warm item distribution, but progressively shift toward and eventually align with the overall warm item embedding distribution.}
%     \label{fig:15-subfigs}
% \end{figure*}

\subsection{Hyper-parameter Analysis (RQ3)}
\label{hyper-parameter-analysis}
We conduct several experiments to examine the effects of two critical hyper-parameters on model performance. 

\subsubsection{Embedding size} As shown in Fig. \ref{fig:embedding_size}, the model performance continues to improve with the increasing embedding size until the rate of improvement diminishes at size 32. Increasing the embedding dimensionality from 32 to 64 results in a proportional growth in parameters, but it yields only marginal performance improvements. Our empirical analysis demonstrates that FedUTR achieves the optimal trade-off between performance and efficiency when employing 32-dimensional embeddings, as further scaling the embedding size beyond this point yields diminishing returns relative to computational overhead.
\begin{figure}[t]
\centering
\includegraphics[width=\columnwidth]{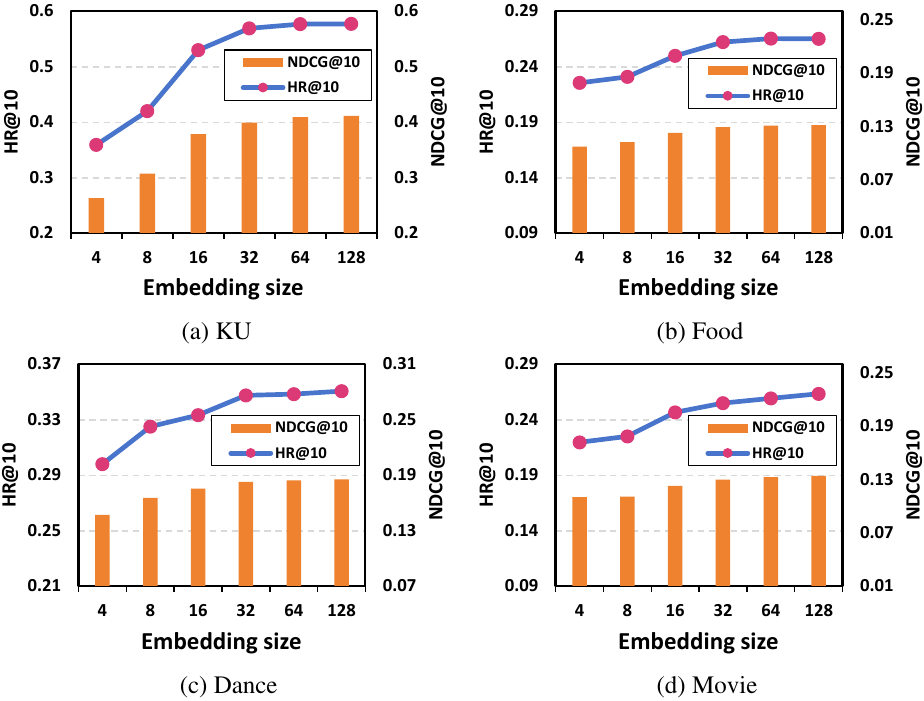} % 子图路径
\caption{Effect of embedding size.}
\label{fig:embedding_size}
\end{figure}

\begin{figure}[t]
\centering
\includegraphics[width=\columnwidth]{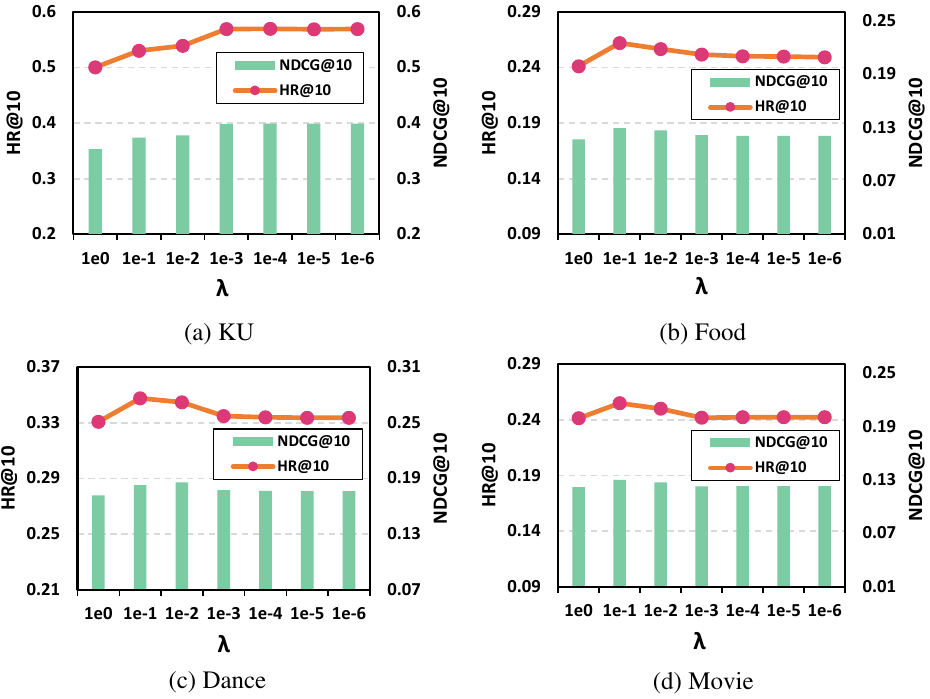} % 子图路径
\caption{Effect of regularization coefficient $\lambda.$}
\label{fig:lambda}
\end{figure}

% \begin{figure*}[t]
%     \centering

%     % 第一行
%     \subfloat[KU HR@10]{
%             \includegraphics[width=0.24\textwidth]{./image/Convergence-KU-HR@10.pdf}
%     }
%     \subfloat[Food HR@10]{
%         \includegraphics[width=0.24\textwidth]{./image/Convergence-Food-HR@10.pdf}
%     }
%     \subfloat[Dance HR@10]{
%         \includegraphics[width=0.24\textwidth]{./image/Convergence-Dance-HR@10.pdf}
%     }
%     \subfloat[Movie HR@10]{
%         \includegraphics[width=0.24\textwidth]{./image/Convergence-Movie-HR@10.pdf}
%     }
    
%     \\ % 换行
    
%     % 第一行
%     \subfloat[KU NDCG@10]{
%         \includegraphics[width=0.24\textwidth]{./image/Convergence-KU-NDCG@10.pdf}
%     }
%     \subfloat[Food NDCG@10]{
%         \includegraphics[width=0.24\textwidth]{./image/Convergence-Food-NDCG@10.pdf}
%     }
%     \subfloat[Dance NDCG@10]{
%         \includegraphics[width=0.24\textwidth]{./image/Convergence-Dance-NDCG@10.pdf}
%     }
%     \subfloat[Movie NDCG@10]{
%         \includegraphics[width=0.24\textwidth]{./image/Convergence-Movie-NDCG@10.pdf}
%     }
% \caption{
% Convergence behavior of FedUTR and baselines on four datasets during the training process. The top row shows the convergence curves of HR@10, and the bottom row reports NDCG@10.
% }
% \label{fig:convergence}
% \end{figure*}

\subsubsection{Regularization strength $\lambda$} Fig. \ref{fig:lambda} demonstrates the impact of different $\lambda$ on model performance. The results reveal that (1) across the Food, Dance, and Movie datasets, model performance initially improves, then degrades as $\lambda$ decreases, and the highest scores are achieved at $\lambda$ = 0.1. (2) On the KU dataset, model performance monotonically improves with decreasing $\lambda$ values until reaching a saturation point at $\lambda$ = 0.001. We analyze performance discrepancies by examining intrinsic characteristics of datasets. The regularization term is introduced to prevent model overfitting. The KU dataset exhibits significantly higher average user interactions (Avg.I=9.11) compared to the other three datasets (Avg.I $\approx$7), thereby inherently reducing overfitting risks due to its richer interaction density. The remaining three datasets, with statistically homogeneous lower average interaction counts, demonstrate consistent performance trends across $\lambda$ values. (3) All datasets exhibit suboptimal performance when $\lambda$ = 1. This phenomenon arises from the optimization dynamics: the recommendation loss dominates the gradient updates, while the L1 regularization term primarily serves as an auxiliary mechanism to prevent overfitting. When $\lambda \geq 1$, the L1 term's influence surpasses that of the recommendation loss, significantly skewing the gradient update trajectory.

\subsection{Plug-and-Play Compatibility Verification (RQ4)}

The URM can be used as a plug-and-play component that provides universal item representations, which can be seamlessly integrated into existing FR models. We conducted experiments to validate the compatibility and effectiveness of URM when integrated with existing FR models. Specifically, we select three representative models (FCF, FedNCF, and FedRAP) as backbones, comparing their original results with URM-enhanced variants. As shown in Table \ref{tab:Plug_results}, the empirical evidence demonstrates significant performance gains across all backbones when URM is incorporated. Notably, even when all backbones have integrated URM, FedUTR still maintains significant performance superiority.

% \subsection{Convergence Validation}
% \begin{figure*}[t]
%     \centering

%     % 第一行
%     \subfloat[KU HR@10]{
%             \includegraphics[width=0.24\textwidth]{./image/Convergence-KU-HR@10.pdf}
%     }
%     \subfloat[Food HR@10]{
%         \includegraphics[width=0.24\textwidth]{./image/Convergence-Food-HR@10.pdf}
%     }
%     \subfloat[Dance HR@10]{
%         \includegraphics[width=0.24\textwidth]{./image/Convergence-Dance-HR@10.pdf}
%     }
%     \subfloat[Movie HR@10]{
%         \includegraphics[width=0.24\textwidth]{./image/Convergence-Movie-HR@10.pdf}
%     }
    
%     \\ % 换行
    
%     % 第一行
%     \subfloat[KU NDCG@10]{
%         \includegraphics[width=0.24\textwidth]{./image/Convergence-KU-NDCG@10.pdf}
%     }
%     \subfloat[Food NDCG@10]{
%         \includegraphics[width=0.24\textwidth]{./image/Convergence-Food-NDCG@10.pdf}
%     }
%     \subfloat[Dance NDCG@10]{
%         \includegraphics[width=0.24\textwidth]{./image/Convergence-Dance-NDCG@10.pdf}
%     }
%     \subfloat[Movie NDCG@10]{
%         \includegraphics[width=0.24\textwidth]{./image/Convergence-Movie-NDCG@10.pdf}
%     }
% \caption{
% Convergence behavior of FedUTR and baselines on four datasets during the training process. The top row shows the convergence curves of HR@10, and the bottom row reports NDCG@10.
% }
% \label{fig:convergence}
% \end{figure*}

\subsection{Convergence Validation (RQ5)}
\begin{figure}[t]
    \centering
    \subfloat[\footnotesize HR@10]{
        \includegraphics[width=0.24\textwidth]{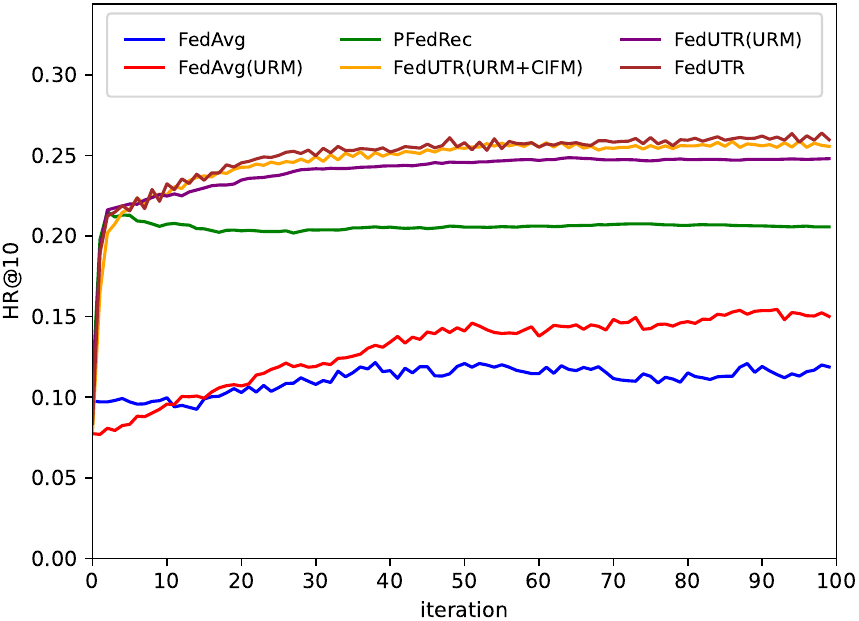}
    }
    \subfloat[\footnotesize NDCG@10]{
        \includegraphics[width=0.24\textwidth]{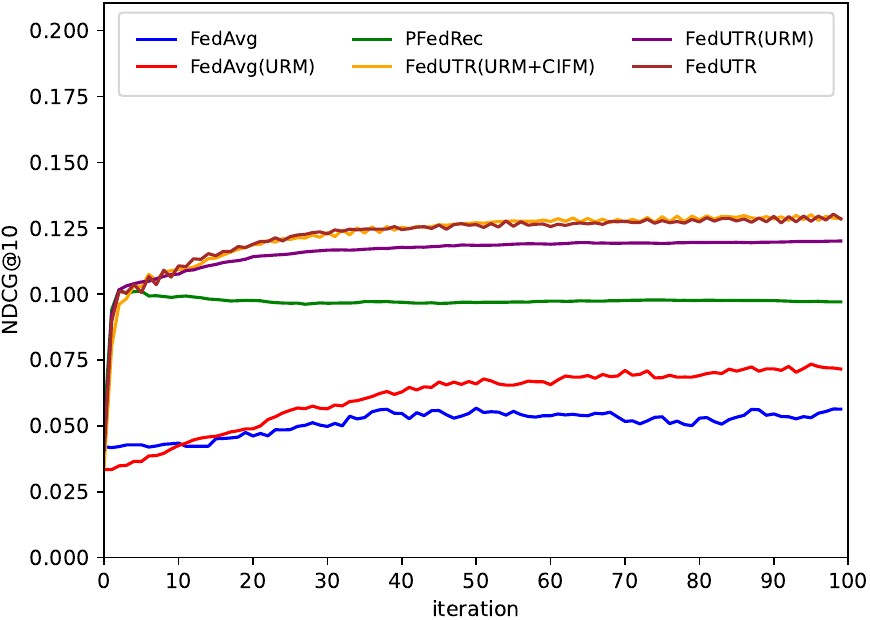} 
    }
\caption{
Convergence behavior of FedUTR and baselines on the Food dataset during the training process. 
}
\label{fig:convergence}
\end{figure}

To validate the convergence of FedUTR, we visualize its recommendation performance on the Food dataset in terms of HR@10 and NDCG@10 during the training process. We select FedAvg and PFedRec as two representative baselines, where FedAvg represents the earliest federated optimization paradigm, and PFedRec serves as a strong state-of-the-art baseline in our experiments.
As shown in Fig. \ref{fig:convergence}, FedUTR consistently converges on the Food dataset, exhibiting stable and monotonic improvements in both HR and NDCG. Compared to FedAvg and PFedRec, FedUTR converges to a better performance level while maintaining a comparable convergence speed. Specifically, URM contributes most significantly to improving convergence efficiency, whereas CIFM and LAM have relatively smaller effects on convergence behavior, which is consistent with the results of the ablation study. 
% Notably, although CIFM and LAM have a limited impact on convergence speed, they play a crucial role in improving final performance.

\subsection{Privacy Protection}
In this section, we conduct robustness evaluations on the privacy protection enhanced FedUTR with Local Differential Privacy (LDP) strategy using the KU and Food datasets. Specifically, we achieve LDP by injecting zero-mean Laplace noise into client-side item embeddings to preserve user privacy.
We consider the noise intensity $\delta = [0.1, 0.2, 0.3, 0.4, 0.5]$. The experimental results in Table \ref{tab:privacy_LDP} demonstrate FedUTR's robustness against noise perturbations of different scales. Despite performance degradation under varying noise intensities, FedUTR maintains remarkable stability, retaining 98.84\% to 99.91\% of its original performance.

\begin{table}[htbp]
\centering
\caption{Experimental results of privacy protection enhanced FedUTR with various noise intensities.}
\label{tab:privacy_LDP}
\setlength{\tabcolsep}{2pt}
\begin{tabular}{lccccccc}
\toprule
{} & \textbf{Intensity}$ \, \delta$  & 0 & 0.1 & 0.2 & 0.3 & 0.4 & 0.5\\
\midrule
\multirow{2}{*}{\textbf{KU}} & HR@10 & \textbf{0.5693} & 0.5688 & 0.5678 & 0.5683 & 0.5688 & 0.5688 \\
\multirow{2}{*}{} & NDCG@10 & \textbf{0.3994} & 0.3984 & 0.3987 & 0.3976 & 0.3980 & 0.3979  \\
\midrule
\multirow{2}{*}{\textbf{Food}} & HR@10 & \textbf{0.2622} & 0.2619 & 0.2594 & 0.2596 & 0.2600 & 0.2597 \\
\multirow{2}{*}{} & NDCG@10 & \textbf{0.1296} & 0.1283 & 0.1284 & 0.1281 & 0.1281 & 0.1282 \\
\bottomrule
\end{tabular}
\end{table}

\section{Conclusion}
\label{section 6}
In this paper, we have revealed the limitations of existing FRs that rely entirely on item ID embeddings under highly sparse scenarios. To address the challenge, we have proposed FedUTR, a novel method that introduces universal representations to complement the shortcomings of ID embeddings. Compared to existing methods, FedUTR not only incorporates modality information but also achieves parameter efficiency, significantly enhancing its potential for practical applications. Meanwhile, our convergence analysis has provided rigorous theoretical guarantees for the effectiveness of the proposed method.
We also introduce a variant, FedUTR-SAR, which incorporates a sparsity-aware residual module to adaptively balance universal and personalized information, providing additional performance improvements. 
Extensive experiments have demonstrated that FedUTR achieves superior performance compared to state-of-the-art baselines. Furthermore, in-depth experiments confirm the compatibility of the URM with existing FR models and its robustness to privacy-preserving aggregation techniques, demonstrating substantial performance improvements in highly sparse scenarios, where traditional methods struggle to perform effectively.

\bibliographystyle{IEEEtran}
\bibliography{paper}

% \newpage

 

% \end{IEEEbiography}



% \vfill

\end{document}